\documentclass{article}

\usepackage[margin=1in]{geometry}
\usepackage[utf8]{inputenc}
\usepackage{amsmath}
\usepackage{amsthm}
\usepackage{graphicx}
\usepackage{hyperref}
\usepackage{xcolor}
\usepackage{float}
\usepackage{cite}

\newtheorem{lemma}{Lemma}
\newtheorem{theorem}{Theorem}
\newtheorem{claim}{Claim}
\newtheorem{remark}{Observation}
\newtheorem{corollary}{Corollary}
\theoremstyle{definition}

\newtheorem{definition}{Definition}

\newcommand{\EPTsum}{\mathsf{EPT\text{-}sum}}

\title{Tight Approximation Bounds on a Simple Algorithm for Minimum Average Search Time in Trees}
\author{Svein Høgemo\thanks{E-mail: svein.hogemo@uib.no}\\ University of Bergen}
\begin{document}

\maketitle

\begin{abstract}
The graph invariant EPT-sum has cropped up in several unrelated fields in later years: As an objective function for hierarchical clustering, as a more fine-grained version of the classical edge ranking problem, and, specifically when the input is a vertex-weighted tree, as a measure of average/expected search length in a partially ordered set. The EPT-sum of a graph $G$ is defined as the minimum sum of the depth of every leaf in an edge partition tree (EPT), a rooted tree where leaves correspond to vertices in $G$ and internal nodes correspond to edges in $G$.

A simple algorithm that approximates EPT-sum on trees is given by recursively choosing the most balanced edge in the input tree $G$ to build an EPT of $G$. Due to its fast runtime, this balanced cut algorithm can be used in practice, and has earlier been analysed to give a 1.62-approximation on trees. In this paper, we show that the balanced cut algorithm gives a 1.5-approximation of EPT-sum on trees, which amounts to a tight analysis and answers a question posed by Cicalese et al. in 2014.
\end{abstract}

\section{Introduction}

Searching in ordered structures is a basic problem in computer science and has seen a lot of attention since the dawn of the field \cite{hib1962,NR1972}. Some special cases are well understood, e.g. the best search strategy in a totally ordered set is a binary search tree, and the problem of optimal search in totally ordered sets given a probability distribution on the elements was treated by Knuth \cite{knu1971} and Hu and Tucker \cite{HT71} in 1971. Also, if the search space is every single subset of elements, then it is known that the optimum search strategy is given by a Huffman tree (see \cite{HT71}). Apart from such specific subcases, little is known except that they are hard problems \cite{CJLM14,LR76}.\\

Searching for an element $x$ in a \emph{partially ordered set} $(A,\preceq)$ can be achieved by a series of queries of the type ``is $x\preceq a$ (for some element $a$)?'', and restricting the next query to the subset $\{a'\in A \mid a'\preceq a\}$ (or $\{a'\in A \mid a'\not\preceq a\}$, depending on the answer to the query). The search finishes when your eligible subset of $A$ contains only one element, which must be $x$.

An interesting subclass of partially ordered sets are ``tree-like posets'', those posets whose Hasse diagrams are rooted trees, or in other words those posets where for any two elements $a,b$, the sets $\{x\mid x\preceq a\}$ and $\{x\mid x\preceq b\}$ either are disjoint, or one is contained in the other. They generalize totally ordered sets (whose Hasse diagrams are paths), but are still much easier to work with than the general case of posets. The average performance of search strategies for tree-like posets have been studied by Laber and Molinaro in \cite{laber2011approximation} and later by Cicalese et al. in \cite{CJLM11,CJLM14}. For these posets, the optimum search strategy with respect to average search time can be shown to be equal to the \emph{EPT-sum} of their Hasse diagrams, a graph invariant recently named in \cite{HBBPT21}, but defined and used several times earlier in other circumstances.

\subsection{EPT's and EPT-sum}

An \emph{edge partition tree}, or EPT for short, of a graph $G$ is a rooted tree where every leaf corresponds to a vertex in $G$, the root corresponds to an edge $e$ in $G$, and each child of the root is itself the root of an EPT of a component of $G\setminus e$. The EPT-sum of $G$ with respect to some EPT $T$ is defined as the sum of the depth of each leaf in $T$; and the EPT-sum of $G$ is defined as the minimum over all EPT's. Viewing the EPT as a certain search strategy in the graph $G$ where, upon each edge $e$, one can query which of the components of $G\setminus e$ contains the wanted vertex, the EPT-sum of $G$ with respect to $T$ is the combined search time of locating every vertex in $G$ with this search strategy. Giving a probability distribution to $V(G)$, the weighted EPT-sum is equal to the expected search time for a vertex pulled at random from the given distribution. When $G$ is a tree, this search strategy is equivalent to the usual search strategy on the tree-like poset that has $G$ as its Hasse diagram. The aforementioned edge rank problem, on the other hand, is interpreted as finding a search strategy with optimal \emph{worst-case} performance on a tree-like poset.\\

Given a tree $G$, we are interested in the following algorithm (the \emph{balanced cut algorithm}) to make an approximately optimal EPT of $G$ (with respect to EPT-sum):
\begin{itemize}
\item Find a balanced cut in $G$, i.e. an edge $e$ that minimizes the size of the biggest component in $G\setminus e$.
\item Make $e$ the root of $T$.
\item Make EPT's of the components of $G\setminus e$ recursively, and make the roots of those trees children of $e$.
\end{itemize}

The balanced cut algorithm is attractive because of its runtime: Finding a balanced edge in a tree takes $O(n)$ time, so making a balanced EPT can trivially be done in $O(n^2)$ time. However, this runtime can be reduced to $O(n \log n)$, as shown in Theorem \ref{thm:runtime}. This makes it practical for any potential applications of searching in trees.\\

For the case of vertex-weighted trees, Cicalese et al.~\cite{CJLM11,CJLM14} show a lower and upper bound on the performance of the balanced cut algorithm of 1.5 and $\varphi$, respectively (where $\varphi$ is the golden ratio $\frac{\sqrt{5}+1}{2}$, approximately equal to 1.62). They ask what the actual performance of this algorithm is, and conjecture that it does provide a 1.5-approximation, matching the lower bound.\\

In this paper we will affirm this conjecture, by proving that the balanced cut algorithm indeed does give a 1.5-approximation of $\EPTsum$ on vertex-weighted trees. Our proof proceeds by constructing a so-called \emph{augmented tree}, built from an optimum EPT by subdividing half of its edges, and showing that it has $\EPTsum$ at most 1.5 times the optimum. Thereafter we iteratively apply operations to this tree that are guaranteed to not increase the $\EPTsum$, eventually arriving at the EPT constructed by the balanced cut-algorithm.

\subsection{Applications of EPT-sum to Hierarchical Clustering}

Employing edge-weighted graphs, a measure equal to EPT-sum was considered by Dasgupta in \cite{Das16} as a suitable objective function for measuring the quality of hierarchical clusterings of similarity graphs (where large weights mean high similarity and vice versa). In this paper and a string of follow-ups \cite{CC17,CKM+19,RoyPokutta}, several attractive properties of EPT-sum as an objective function were highlighted, and approximation algorithms for EPT-sum were found. In light of this, a balanced edge of a tree is just a special case of the \emph{sparsest cut} of an edge-weighted graph, a partition $(A,B)$ of $V(G)$ that minimizes the ratio $(\sum_{e\in E(G[A,B])} x_e)/(|A|\cdot|B|)$ (where $x_e$ is the weight of the edge $e$). 

Charikar and Chatziafratis show that the balanced cut algorithm gives an 8-approximation of the $\EPTsum$ of edge-weighted trees (and other graph classes for which an optimal sparsest cut can be found in polynomial time, like planar graphs, see~\cite{abboud2020new} for more information). Whether this is the real approximation ratio is, however, unknown.

\subsection{Related Parameters}

EPT's were introduced by Iyer et al. in \cite{IYER199143} along with the \emph{edge rank problem}; the problem of finding an EPT of a graph with lowest possible height. EPT's are a variation of \emph{vertex partitioning trees} or VPT's, defined similarly but where every node in the tree corresponds to a vertex in the graph. VPT's and the \emph{node rank problem} (also called \emph{tree-depth}) of finding a VPT of lowest height have applications in such disparate areas as VLSI chip design \cite{leiserson1980area}, sparse matrix multiplication \cite{BP1993,liu1990}, and structural graph theory \cite{Nesetril015,NESETRIL20061022}. The edge rank problem was rediscovered in terms of searching in tree-like posets by Ben-Asher et al.~\cite{BFN99}, who also found a polynomial-time algorithm for the problem (although polynomial- and linear-time algorithms were already known for edge ranking, see \cite{de1995optimal,LamYue2001}) The equivalence between edge rank of trees and searching in tree-like posets was first pointed out by Dereniowski in \cite{Der2008}.\\

There are also several relevant results on the related parameter of \emph{VPT-sum}; the VPT-sum of a graph $G$ with respect to a \emph{vertex partition tree} (\emph{VPT}) $T$ is the sum of depth plus one of every node in the VPT. Optimizing VPT-sum is, on trees, equivalent to optimizing a different search model where one has access to an oracle which, for each vertex $v\in V(G)$, answers which of the components of $G\setminus v$ (or, potentially, $\{v\}$ itself) contains the vertex one is searching for. For this parameter, the corresponding strategy of recursively choosing the most balanced vertex (i.e. the \emph{centroid}) to build a VPT, was recently shown by Berendsohn et al.~\cite{berendsohn2022fast} to give exactly a 2-approximation of VPT-sum -- both for weighted and unweighted trees. In addition, there exists a PTAS for VPT-sum on trees~\cite{BK22}.\\

EPT-sum and the related parameters that are treated in \cite{HBBPT21} are all NP-hard to calculate exactly on general, unweighted graphs. In contrast, on unweighted trees, the complexity of EPT-sum and VPT-sum are unknown. For vertex-weighted trees, the complexity of calculating VPT-sum is also, to the best of our knowledge, unknown, but it is NP-hard to calculate EPT-sum~\cite{CJLM11}. One should note that in the NP-hardness reductions employed in~\cite{CJLM11}, the trees have exponential weights. In \cite{HBBPT21}, an equivalence between EPT-sum on vertex-weighted trees and unweighted trees was demonstrated. So, if EPT-sum should turn out to be NP-hard on trees where every vertex has polynomial weight, then it is also NP-hard on unweighted trees.\\

Independently of the result itself, we believe that there might be some interest in the proof technique used here. Specifically, the trick of building an augmented tree that has cost 1.5 times the optimum has, to our knowledge, not been used before.

\subsection{Organization}

The rest of the paper is organized as follows: in Section 2 we give basic definitions. We also show that there is an algorithm for making balanced EPT's that run in $O(n \log n)$ time. In Section 3 we prove the main result of this paper, namely the fact that a balanced EPT has EPT-sum at most 1.5 the optimum. Finally, in Section 4 we restate some of the most relevant results and open problems regarding computing and approximating EPT-sum on trees.

\section{Preliminaries}

\subsection{Basic notions}

In general, we follow the notation established in \cite{HBBPT21}. The vertex set and edge set of a graph $G$ is denoted $V(G)$ and $E(G)$, respectively. Given vertices $S\subseteq V(G)$, $G[S]$ signifies the induced subgraph of $G$ on $S$. As both the input graphs and the output data structures are trees, we employ the following conventions to avoid confusion: An unrooted tree that is given as input is denoted $G$, while a rooted tree generated as output is denoted $T$, and its vertices are called \emph{nodes}. For a node $x\in V(T)$, $T[x]$ signifies the rooted subtree of $T$ rooted in $x$. The set of \emph{leaves} in $T$ is denoted $L(T)$, a node which is not a leaf is called an \emph{internal node}. A \emph{full binary tree} is a binary tree where every internal node has exactly two children.

The \emph{depth} of a node $x$ in a rooted tree $T$ is equal to the distance from the root $r$ to $x$ in $T$. For example, $r$ itself has depth 0. The \emph{height} of $T$ is the maximum depth out of any node in $T$. The nodes on the path from $r$ to $x$ (including both $r$ and $x$) are the \emph{ancestors} of $x$, and $x$ is a \emph{descendant} of all these nodes.

\subsection{Edge Partition Trees}

Given a connected graph $G$, an \emph{edge partition tree}, or EPT for short, of $G$ is a rooted tree with $|V(G)|$ leaves and $|E(G)|$ internal nodes, that can be defined inductively as follows:
\begin{itemize}
\item The EPT of a graph with a single vertex is a tree with a single node.
\item For a bigger graph $G$, the root $r$ of any EPT $T$ corresponds to an edge $e\in E(G)$, and $r$ has at most two children $c_1,c_2$; one for each component of $G\setminus e$. $T[c_1]$ (and $T[c_2]$, if it exists) are, in turn, EPT's of the component(s) of $G\setminus e$.
\end{itemize}

\begin{figure}[ht!]
\centering
\includegraphics[width=0.6\textwidth]{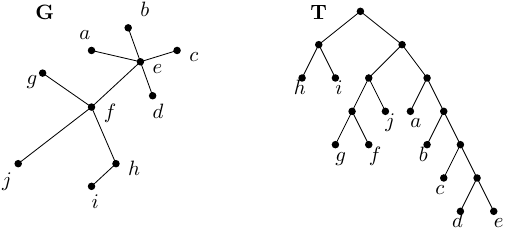}
\caption{An unweighted tree $G$ and an EPT $T$ of $G$. Adding up the depth of each leaf, one sees that $\EPTsum(G,T) = 39$. This is not optimal for $G$; a better EPT can be made by making the edge $ef$ (which incidentally is also the most balanced edge in $G$) root.}
\end{figure}

An EPT is, as such, a binary tree. A graph $G$ is a tree if and only if any (and every) EPT of $G$ is a full binary tree (since in that case, every edge of $G$ is a bridge). Since we focus on trees as input, we assume that EPT's are full binary trees.

There is a bijection from the vertex set of $G$ to the leaves of $T$, and likewise from the edge set of $G$ to the internal vertices of $T$. To simplify notation somewhat, we will not notate these bijections and rather say that each leaf of $T$ \emph{is} a vertex of $G$ and each internal node of $T$ \emph{is} an edge of $G$. For an internal node $v\in T$, we define $G[v]$ as a shorthand for $G[L(T_v)]$, the (connected) subgraph of $G$ induced by the rooted subtree under $v$.

\subsection{EPT-sum and Searching in Trees}

The focus of this paper is on the graph measure called EPT-sum. We give the most general definition, where $G$ is equipped with both edge and vertex weights:

Given a graph $G$ with vertex weights $W = \{w_v\mid v\in V(G)\}$ and edge weights $X = \{x_e\mid e\in E(G)\}$, and an EPT $T$ of $G$, we define the EPT-sum of $G$ with respect to $T$ as follows:

\begin{definition}[EPT-sum]\label{def:EPT1}
$\EPTsum(G,T)$ is equal to $$\sum_{e\in E(G)} x_e\cdot \left(\sum_{v\in L(T[e])} w_v\right)$$
If $G$ is not vertex-weighted, we can replace the sum of vertex-weights by $|L(T[e])|$.

The EPT-sum of $G$, $\EPTsum(G)$ is the minimum of $\EPTsum(G,T)$ over all EPTs $T$.
\end{definition}

\begin{definition}[EPT-sum, alternative def.]\label{def:EPT2}
We can easily verify through reordering of terms that EPT-sum can be equivalently defined as follows:
$$\sum_{v\in V(G)} w_v\cdot \left(\sum_{e\in \mathrm{\:ancestors\:of\:} v} x_e\right)$$
If $G$ is not edge-weighted, we can replace the sum of edge-weights by the depth of $v$ in $T$.
\end{definition}

The graphs we will treat in this paper are vertex-weighted, but not edge-weighted. Therefore the formula $\sum_{v\in V(G)} (w_v\cdot dist_T(r,v))$, where $r$ is the root of $T$, is a natural choice to use. It is also this formulation that makes EPT-sum and similar measures attractive for measuring the performance of search trees.

We must note that the equivalence between EPT-sum and average search time only holds if the Hasse diagram actually is a rooted tree; otherwise these are two different problems. This is evident from the fact that the sparsest cut algorithm from \cite{CC17} provides a $O(\sqrt{\log n})$-time approximation of EPT-sum in edge-weighted graphs, while it was proven in \cite{CJLM11} that searching in posets (i.e. vertex-weighted DAG's) has no $o(\log n)$-approximation unless every problem in NP admits a quasipolynomial algorithm (see also \cite{LN04}).

\subsection{A Fast Balanced Cut Algorithm}

Before moving on to the main result, we show that finding a balanced EPT can be done more quickly than the basic algorithm outlined in the introduction. This has probably been noted beforehand; however we have been unable to find such a result in the literature.

\begin{theorem}\label{thm:runtime}
Given a tree $G$, one can compute a balanced EPT of $G$ in time $O(n \log n)$.
\end{theorem}
\begin{proof}
It is a well-known fact that every tree on $n$ vertices contains a \emph{centroid} vertex: a vertex where, if you remove it, the remaining subtrees all contain at most $n/2$ vertices. We derive the following two claims:

\begin{claim}\label{cl:centroid1}
Let $G$ be a tree, and let $c$ be a centroid in $G$. Every balanced edge in $G$ is incident on $c$.
\end{claim}
\emph{Proof of claim.} Any edge $e$ that is not incident on $c$ is contained in some component $C$ of $G\setminus c$. Therefore one of the components of $G\setminus e$ is strictly contained within $C$. But $C$ is a smallest component of $G\setminus e_C$, where $e_C$ is the edge connecting $c$ with $C$. Therefore $e_C$ is strictly more balanced than $e$.

\begin{claim}\label{cl:centroid2}
Let $G$ be a tree, and let $e$ be a balanced edge in $G$. Furthermore, let $C$ be a biggest component of $G\setminus e$, and let $c$ be the endpoint of $e$ lying in $C$. If $|V(C)| \geq \frac{2n}{3}$, then $c$ is a centroid in $C$.
\end{claim}
\emph{Proof of claim.} From the previous claim, we find that $c$ must be a centroid in $G$. Let $D$ be the small component of $G\setminus e$ (that is, $D = G\setminus C$); by assumption, $|V(D)| \leq \frac{1}{3}|V(G)|$. Furthermore, for any component $D'$ of $G\setminus c$, $|V(D')| \leq |V(D)|$; otherwise, $e$ would not be the most balanced edge in $G$. Also, every component of $G\setminus c$, except $D$, clearly is a component of $C\setminus c$ as well. As such, every component of $C\setminus c$ has at most half as many vertices as $C$ itself, which implies that $c$ is centroid in $C$.\\

These two claims together show that you will always find a balanced edge by a centroid, and furthermore, if the largest component when that edge is removed has more than $\frac{2n}{3}$ vertices, we can find a balanced edge of the large component by the same centroid. Therefore, we can consider the following algorithm for building a balanced EPT of $G$:
\begin{enumerate}
\setcounter{enumi}{-1}
\item If $G$ consists of a single leaf $x$; output $x$ and terminate. Else, move on to step 1.
\item Locate a centroid vertex $c$ in $G$ and let $k$ denote the degree of $c$. Put all edges $cv_1,\ldots,cv_k$ incident on $c$ into a list.
\item Sort each edge $cv_i$ in decreasing order of $|C_i|$, where $C_i$ is the component of $G\setminus c$ containing $v_i$. (From now on, we assume that the labels $1,\ldots,k$ refers to this ordering of the edges.)
\item Find the minimum $s\leq k$ such that $(n-\sum_{i=1}^s |C_i|) < 2|C_s|$.
\item Make a caterpillar-shaped partial EPT $T$ with $cv_1$ as root and $cv_2,\ldots,cv_s$ downward.
\item Build onto $T$ by recursing on the subtrees left by removing $e_1,\ldots,e_s$ from $G$, output the finished EPT $T$ and terminate.
\end{enumerate}

The correctness of the above algorithm follows from Claims \ref{cl:centroid1} and \ref{cl:centroid2}, since all the edges $cv_1,\ldots,cv_s$ are guaranteed to be balanced edges of progressively smaller subtrees. We now argue that the algorithm runs in time $O(n \log n)$. It is a recursive algorithm, so a run of the algorithm can be visualized as a rooted tree, where the root represents $G$, and the children of the root represent the components of $G\setminus\{cv_1,\ldots,cv_s\}$ on which the algorithm recurses. We observe two things:
\begin{itemize}
\item Since every component of $G\setminus\{cv_1,\ldots,cv_s\}$ has size less than $\frac{2n}{3}$, the height of the recursion tree is $O(\log n)$.
\item The nodes at a particular ``level'' of the recursion tree (i.e. a set of nodes that lie equally far from the root) induce a partition of some subgraph of $G$ into subtrees.
\end{itemize}

From these two observations, we see that if the work done on a single node of the recursion tree can be achieved in linear time, then the algorithm as a whole runs in $O(n \log n)$ time, since the work on all the nodes on a single level of the recursion tree can be done in $O(n)$ time, and there are $O(\log n)$ levels. So this is what we are going to argue. The work done at each node of the recursion tree consists of steps 1 through 4 of the algorithm. Step 1 (finding a centroid and the edges incident on it) can be done in linear time by a standard algorithm that roots $G$ in an arbitrary vertex and calculates the size of each rooted subtree. Steps 3 (finding $s$) and 4 (building the partial EPT), as well as generating the subtrees to recurse over, are easy to see that we can do in linear time.

The runtime of step 2 (sorting the edges) is the one that we need to argue about. From step 1, we already know the size of every subtree $C_1,\ldots,C_k$, so we do not have to recompute those. The number $k$ of edges incident on a vertex $c$ can obviously be $\Omega(n)$ in the worst case. However, the number of different values that the numbers $|C_1|,\ldots,|C_k|$ can take is $O(\sqrt{n})$, because all components are disjoint and $\sum_{i=1}^k |C_i| = n-1$ (specifically, the number of different values must be upper bounded by the largest integer $t$ such that the triangular number $1+2+\ldots+t$ is smaller than $n$). Therefore it is possible to sort the edges $cv_1,\ldots,cv_k$ in time (and space) $\Theta(n)+\Theta(\sqrt{n}\log(n)) = O(n)$ by the following subroutine:
\begin{itemize}
\item Make an array $A$ of length $\lfloor\frac{n}{2}\rfloor$, where each element $A[1],\ldots,A[\lfloor\frac{n}{2}\rfloor]$ is a stack.
\item For each $1\leq i\leq k$, push $vc_i$ to the stack $A[|C_i|]$.
\item Fill a second array $B$ with the set $\{1\leq j\leq \lfloor\frac{n}{2}\rfloor \mid A[j]\text{ is non-empty}\}$, and sort $B$ in decreasing order. As we noted, the length of $B$ is $O(\sqrt{n})$.
\item Now, we can quickly fetch each edge $cv_i$ in decreasing order of $|C_i|$.
\end{itemize}
This completes the proof that the above algorithm runs in time $O(n \log n)$.
\end{proof}

\section{Balanced EPT's have EPT-sum at most 1.5 the Optimum}

We will go through the proof with unweighted graphs in mind, but note that the proof is agnostic as to whether the vertices have weights (see Corollary \ref{cor:weights}). The proof therefore also works for vertex-weighted trees.

The first step is building the tree that has $\EPTsum$ at most 1.5 times as high as the optimal tree.

\begin{definition}[Augmented tree]
Given a full binary tree $T$, the \emph{augmented tree} of $T$, denoted $aug(T)$, is constructed in the following manner: For any internal node $v\in V(T)$ with children $c_l,c_r$, choose one child (say, $c_r$) with the property that $|L(T[c_r])|\leq |L(T[c_l])|$ (this is obviously true for at least one of $c_l,c_r$). Then, subdivide the edge $vc_r$ once.
\end{definition}

If $T$ is an EPT of a tree $G$, then $aug(T)$ is not an EPT of $G$, since it has more internal nodes than there are edges in $G$. Nonetheless, we define $\EPTsum(aug(T))$ to be the sum of depths of leaves in $aug(T)$.

\begin{lemma}\label{lemma:augmented}
For any full binary tree $T$, $\EPTsum(aug(T)) \leq \frac{3}{2}\EPTsum(T)$.
\end{lemma}

\begin{proof}
The proof goes by induction. For the base case, we observe that a full binary tree with one leaf has $\EPTsum$ equal to 0, and no edges to augment.

For the inductive step, assume that the lemma holds for any full binary tree with at most $n-1\geq 1$ leaves, and let $T$ be an arbitrary full binary tree with $n$ leaves, and $aug(T)$ the augmented tree of $T$. Furthermore, let $r$ be the root of $T$, and $c_l,c_r$ the children of $r$. Note that $c_l$ and $c_r$ must exist, by the assumption that $T$ is a full binary tree with at least two leaves. Also, note that $aug(T)[c_l]$ (resp. $aug(T)[c_r]$) is the augmented tree $aug(T[c_l])$ (resp. $aug(T[c_r])$).

W.l.o.g. assume that $c_r$ is the child of $r$ such that the edge $rc_r$ is subdivided in $aug(T)$. Then we know that $|L(T[c_r])| \leq \frac{n}{2}$. Also, both $T[c_l]$ and $T[c_r]$ have at most $n-1$ leaves.

By Definition \ref{def:EPT1}, $$\EPTsum(T) = \EPTsum(T[c_l])+\EPTsum(T[c_r])+n,$$ while $$\EPTsum(aug(T)) = \EPTsum(aug(T[c_l]))+\EPTsum(aug(T[c_r]))+n+|L(T[c_r])|,$$ which we have seen to be upper-bounded by $$\EPTsum(aug(T[c_l]))+\EPTsum(aug(T[c_r]))+\frac{3}{2}n.$$ Therefore $$\frac{\EPTsum(aug(T))}{\EPTsum(T)} \leq \frac{\EPTsum(aug(T[c_l]))+\EPTsum(aug(T[c_r]))+\frac{3}{2}n}{\EPTsum(T[c_l])+\EPTsum(T[c_r])+n}.$$
The right hand side fraction is a mediant of the three fractions $\frac{\EPTsum(aug(T[c_l]))}{\EPTsum(T[c_l])}$, $\frac{\EPTsum(aug(T[c_r]))}{\EPTsum(T[c_r])}$ and $\frac{3}{2}$. By the induction hypothesis, all these fractions are upper-bounded by $\frac{3}{2}$, from which we conclude that $\frac{\EPTsum(aug(T))}{\EPTsum(T)} \leq \frac{3}{2}$.
\end{proof}

One should note that this Lemma also holds for EPT's of vertex-weighted trees; this is easily seen by replacing $\frac{n}{2}$ with $\frac{\sum_{v\in V(G)}w_v}{2}$ in the formula.\\

We introduce one additional notion that will come up in the proof of Theorem \ref{thm:approx}:

\begin{definition}[Splitting]\label{def:splitting}
Given a tree $G$ and an EPT $T$ of $G$, and an edge $uv\in E(G)$, the \emph{splitting} of $T$ along $e$ is a pair of rooted trees $T^u,T^v$, which are EPT's of the components of $G\setminus e$, $G_u$ and $G_v$ respectively. $T^u$ and $T^v$ are defined as follows: $L(T^u)$ and $L(T^v)$ are equal to $V(G_u)$ and $V(G_v)$ respectively, and for any node $x\in V(T^u)$ (resp. $V(T^v)$), its parent is equal to the lowest ancestor of $x$ in $T$ whose corresponding edge lies within $G_u$ (resp. $G_v$). If $x$ has no such ancestor, then it becomes the root of $T^u$ (resp. $T^v$).
\end{definition}

\begin{remark}\label{rmk:splitting}
$T^u$ and $T^v$ are, indeed, EPT's of $G_u$ and $G_v$.
\end{remark}
\begin{proof}
Let $T'$ be a subtree of $T$ that includes the root $r$. It is evident that $T'$ induces a partition of $G$ into connected subgraphs, corresponding to the rooted subtrees in $T$ under each leaf of $T'$. We observe that there is maximally one such subgraph of $G$ that intersects both $G_u$ and $G_v$; namely, the subgraph containing $uv$. (Conversely, if $T'$ contains $uv$, then there is no such subgraph of $G$.) Therefore, any internal node $x$ in $T$ whose corresponding edge lies in, say, $G_u$, has exactly two highest descendants whose edges (or vertices, if they are leaves) also lie within $G_u$. These two descendants must in turn correspond to the components of $G_u[L(T_x)]\setminus x$. Therefore, $T^u$ and $T^v$ are EPT's of $G_u$ and $G_v$.
\end{proof}

\begin{lemma}\label{lem:splitting}
Let $G$ be a tree and $T$ an EPT of $G$, and let $T^u,T^v$ be a splitting of $T$ along an edge $uv\in E(G)$. Then, $\EPTsum(G_u,T^u)+\EPTsum(G_v,T^v) < \EPTsum(G,T)$.
\end{lemma}
\begin{proof}
From Definition \ref{def:EPT1}, we see that it is enough to prove that for any internal node $x$ in, say, $T^u$, $L(T^u_x) \leq L(T_x)$. Due to the way $T^u$ is constructed, every leaf in $T^u_x$ is also a leaf in $T_x$, and the claim immediately follows. The strict inequality stems from the fact that $L(T_{uv}) > 0$, but $uv$ is not present in either $T^u$ or $T^v$.
\end{proof}

\begin{theorem}\label{thm:approx}
The balanced cut-algorithm gives a 1.5-approximation of EPT-sum on trees.
\end{theorem}
\begin{proof}
Let $G$ be an unrooted tree, with $T^*$ an optimal EPT of $G$, and $T'$ an EPT of $G$ given by the balanced cut-algorithm. We want to show that $\EPTsum(T') \leq \EPTsum(aug(T^*))$. We achieve this through a procedure that incrementally modifies $aug(T^*)$ into $T'$.

We refer to an internal node $v$ in $T_i$ as \emph{correctly placed} when (i) $v$ corresponds to a most balanced edge in $G[(T_i)_v]$, and (ii) every ancestor of $v$ is also correctly placed. Clearly, in an EPT output by the balanced cut algorithm, every node is correctly placed, and vice versa. During the procedure, we generate a series of trees $T_0,T_1,\ldots,T_t$ where $T_0 = aug(T^*)$, $T_t = T'$ and, for every $0\leq i<t$, $T_{i+1}$ is made from $T_i$ by way of the following algorithm:

\begin{itemize}
\item If $T_i = T'$: set $t := i$ and halt.
\item Otherwise: select an internal node $v\in V(T_i)$ that has a subdivided edge to one of its children, but where none of its ancestors have such an edge.
\item Modify $T_i[v]$ in a suitable way, according to Cases 1-7 below (depending on where the balanced edge of $G[v]$ is located in $T_i[v]$).
\item Set $T_{i+1}$ equal to the modified tree.
\end{itemize}

For every $0\leq i < t$, the number of correctly placed nodes in $T_{i+1}$ is at least as high as in $T_i$. When a node has been placed correctly in the EPT, the subdivided edge from it to one of its children is smoothed out; therefore it is gradually modified from an augmented EPT to a non-augmented EPT. The node $v$ that is chosen in each step is a node of minimum depth that is not guaranteed to be correctly placed, since we remove the subdivision node underneath each correctly placed node. Therefore, after replacing $v$ with a node corresponding to a balanced edge, the replacement node is now correctly placed, increasing the number of correctly placed nodes.

The proof hinges on being able to show that the following invariant holds for every $0\leq i\leq t$: $$\EPTsum(G,T_i) \leq \EPTsum(G,aug(T^*))$$ The invariant trivially holds for $i = 0$. The rest of the proof goes by induction; assuming that the invariant holds for some $i$ and showing that it then also must hold for $i+1$. Given the definition $$\EPTsum(G,T) = \sum_{e\in E(G)} \left(\sum_{v\in L(T[e])} w_v\right)$$ we can see that if the modification from $T_i$ to $T_{i+1}$ is constrained to the rooted subtree $T_i[v]$, then $$\EPTsum(G,T_{i+1})-\EPTsum(G,T_i) = \EPTsum(G,T_{i+1}[v])-\EPTsum(G,T_i[v])$$ (this is what Dasgupta~\cite{Das16} refers to as \emph{modularity of cost}).

How to actually modify $T_i$ depends on where the balanced edge is located, and is subject to a lengthy analysis of seven cases, which constitutes the rest of the proof. Each case concentrates on a subtree $T_i[v]$ where $v$ was selected in the manner described above. For visual aid, we include a picture of each case and how that particular tree is modified. We believe that explaining the modification in each case in words does not add any explanatory power over the figures themselves -- specifically, the pictures replace sentences like ``we smooth out these and these edges, perform such and such rotations on the tree, and subdivide those and those edges''. Furthermore, the inferred change in EPT-sum for each case is found by counting up the depth of each leaf in the trees shown in the figures. We will therefore restrict a thorough explanation to Cases 1 and 2.\\

We will employ the following conventions throughout the case analysis: The balanced edge in question is denoted $e$, and its corresponding node in $T_i[v]$ is located in the subtree on the left hand side of the root (except in Case 1). The leftmost and rightmost subsets of leaves (corresponding to subsets of $V(G)$) are called $A$ and $B$, respectively, and the root is called $v$, its left child $l$, and $l$'s right child is called $b$ (except in Cases 1-2). Lastly, $A\cup B$ never forms a connected subgraph of $G$ (again, except in Case 1, where $A\cup B$ is the whole of $L(T_i[v])$). Note that since $e$ lies in the left subtree, the right subtree (containing $B$) must be strictly smaller, and can therefore safely be assumed to be the one having a subdivided edge from the root in $T_i[v]$.

Given these conventions, one should take care to get convinced that in each case, all rooted subtrees (both in the original tree and the modified one) induce connected subgraphs of $G$; this is important, as it is necessary in order for the arguments to hold water.

In the figures, pink dots signify subdivision nodes in the augmented tree. A pink line between two edges signifies that one of those edges is subdivided (whichever leads to the subtree with the fewest leaves), but we do not know which.\\

\textbf{Case 1:} $v = e$.\\

\begin{figure}[ht!]
\centering
\includegraphics[width=0.5\textwidth]{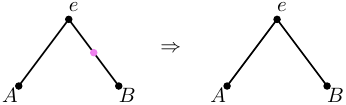}
\caption{Case 1.}
\end{figure}

This is certainly the easiest case. Since $e$ is already correctly placed, the only modification we do is smoothing the subdivided edge from the root. It is trivial to show that $\EPTsum(G,T_{i+1}) \leq \EPTsum(G,T_i)$ and the invariant holds.\\

\textbf{Case 2:} $l=e$.\\

\begin{figure}[ht!]
\centering
\includegraphics[width=0.5\textwidth]{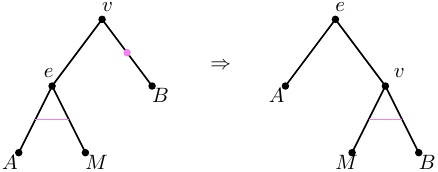}
\caption{Case 2.}
\end{figure}

Both here and in the following cases, we need to remember that $A$ is defined to be separated from $B$ by $M$, therefore $G[B\cup M]$ is connected; as previously mentioned, this is necessary in order for the modified tree to be an EPT of $G[v]$.\\

In this case, we perform a binary tree rotation on $T_i[v]$ around $v$ and $e$, and smooth the subdivided edge from $v$ to $B$, decreasing the EPT-sum by $|B|$ (but the act of rotating around $v$ and $e$ simultaneously increases EPT-sum by $|B|$, so these cancel each other out). We further smooth the edge from $e$ to the smallest subtree of $A$ and $M$ (decreasing EPT-sum by $\min(|A|,|M|)$, and all leaves in $A$ decrease one in depth, decreasing the EPT-sum by $|A|$ (due to the rotation around $v$ and $e$), but we subdivide the edge from $v$ to the smallest subtree out of $M$ and $B$, increasing the EPT-sum with $\min(|B|,|M|)$.

Taking all this into account, the change in EPT-sum therefore simplifies to $$\min(|B|,|M|)-$$ $$(|A|+\min(|A|,|M|))$$ where all increase is collected on the top line and decrease on the bottom line. We observe that the inequality $|A| \geq |B|$ must hold, as otherwise $v$ would lead to a more balanced cut than $e$. It is a basic observation that we use throughout the proof, that if $a\leq b$ and $c\leq d$, then $\min(a,c)\leq \min(b,d)$. In this case, we get $\min(|B|,|M|)\leq |B|\leq |A|$. Therefore the change in EPT-sum is non-positive and $\EPTsum(G,T_{i+1}) \leq \EPTsum(G,T_i)$; thus the invariant holds.\\

\textbf{Case 3:} $e$ lies inside $G[A]$.\\

\begin{figure}[ht!]
\centering
\includegraphics[width=0.5\textwidth]{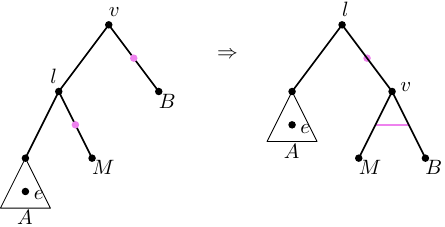}
\caption{Case 3.}
\end{figure}

This is a unique case in the sense that $e$ does not become correctly placed in $T_i[v]$ at this step. Therefore it is important to keep a subdivided edge from the new root $l$ in this case, so that we can balance $T_{i+1}[v]$ at a later step, before moving to its subtrees.

We observe that the inequality $|A| \geq |M|+|B|$ must hold, otherwise $l$ would lead to a more balanced cut than $e$ (this is also why we can be sure that the edge from $l$ to $M$ is subdivided). The change in EPT-sum is $$(|B|+\min(|M|,|B|))-$$ $$|A|.$$ From the aforementioned inequality, the change in EPT-sum is non-positive and therefore $\EPTsum(G,T_{i+1}) \leq \EPTsum(G,T_i)$; thus the invariant holds.\\

In the next three cases (4-6), we assume that $A$ and $B$ lie on different sides of $e$, i.e. that there are connected subgraphs $G[L],G[R]$ where $L\cup R = M$ such that the components of $G[L(T_i[v])]\setminus e$ are $G[A\cup L]$ and $G[B\cup R]$. In the last case, Case 7, we assume the opposite; that $A$ and $B$ are in the same component of $G\setminus e$.\\

\textbf{Case 4:} $e$ is the root of the tree containing $L\cup R$.\\

\begin{figure}[ht!]
\centering
\includegraphics[width=0.5\textwidth]{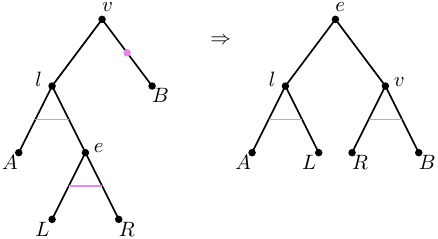}
\caption{Case 4.}
\end{figure}

Here, the change in EPT-sum is $$(\min(|A|,|L|)+\min(|R|,|B|))-$$ $$(\min(|A|,|L\cup R|)+\min(|L|,|R|)+|L|+|R|).$$ As the sum of the top row is clearly at most $|L|+|R|$, the change in EPT-sum is non-positive and $\EPTsum(G,T_{i+1}) \leq \EPTsum(G,T_i)$; thus the invariant holds.\\

\textbf{Case 5:} One of the children of $b$ (containing $e$) contains $R$.\\

\begin{figure}[ht!]
\centering
\includegraphics[width=0.5\textwidth]{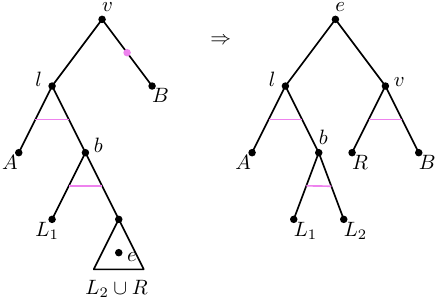}
\caption{Case 5.}
\end{figure}

Both here and in the next case, we look at the splitting of the subtree containing $e$ (in this case $T_i[L_2\cup R]$). By Lemma \ref{lem:splitting}, it is clear that, also when augmented, the two subtrees $T_{i+1}[L_2]$ and $T_{i+1}[R]$ in sum have lower $\EPTsum$ than $T_i[L_2\cup R]$ has. Let us argue in more detail: By Lemma \ref{lem:splitting}, the EPT's $T_{i+1}'[L_2]$ and $T_{i+1}'[R]$ (made by smoothing out all the subdivided edges in $T_{i+1}[L_2]$ and $T_{i+1}[R]$) have, in sum, lower EPT-sum than the EPT $T_i'[L_2\cup R]$ (made by smoothing out the subdivided edges in $T_i[L_2\cup R]$) has. Subdividing the edges again, we see that every subdivided edge from a node $x$ to its child in (say) $T_{i+1}[L_2]$ corresponds to a subdivided edge from $x$ to its child in $T_i[L_2\cup R]$. Therefore augmenting $T_i[L_2\cup R]$ must increase EPT-sum at least as much as augmenting $T_{i+1}[L_2]$ and $T_{i+1}[R]$. So, if we cancel out $\EPTsum(T_{i+1}[L_2]) + \EPTsum(T_{i+1}[R]) - \EPTsum(T_i[L_2\cup R])$, we can infer that the change in EPT-sum is at most $$(\min(|A|,|L|)+\min(|L_1|,|L_2|)+\min(|R|,|B|))-$$ $$(\min(|A|,|L\cup R|)+\min(|L_1|,|L_2\cup R|)+|R|)$$ where we assume $L = L_1\cup L_2$. We see that each term in the top row is at most as high as the corresponding term in the bottom row. As such, the change in EPT-sum is non-positive and $\EPTsum(G,T_{i+1}) \leq \EPTsum(G,T_i)$; thus the invariant holds.\\

\textbf{Case 6:} One of the children of $b$ (containing $e$) contains $L$.\\

\begin{figure}[ht!]
\centering
\includegraphics[width=0.5\textwidth]{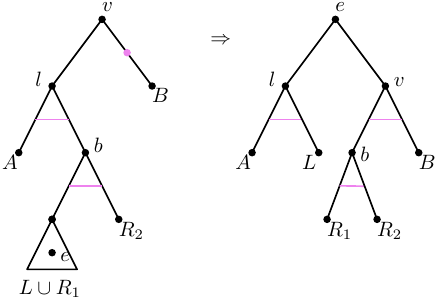}
\caption{Case 6.}\label{fig:case6}
\end{figure}

As in the previous case, we argue by Lemma \ref{lem:splitting} that the sum of the $\EPTsum$ of the splitting $T_{i+1}[L],T_{i+1}[R_1]$ is strictly lower than $\EPTsum(T_i[L\cup R_1])$. Furthermore, we will show that this sum is actually upper-bounded by $\EPTsum(T_i[L\cup R_1])-\min(|L|,|R_1|)$. Look at the root of the subtree $T_i[L\cup R_1]$. Either that root is equal to $e$, in which both $L$ and $R_1$ sit at a lower level in $T_i[v]$ than what is shown in Figure \ref{fig:case6}, or it is not equal to $e$, in which one of the two -- at least the smallest one -- sits at a lower level than what is shown in Figure \ref{fig:case6}. We can therefore add a $\min(|L|,|R_1|)$ to the bottom line of the change in EPT-sum which thus reads $$(\min(|A|,|L|)+\min(|R|,|B|)+\min(|R_1|,|R_2|))-$$ $$(\min(|A|,|L\cup R|)+|L|+\min(|L\cup R_1|,|R_2|)+\min(|L|,|R_1|))$$ where $R = R_1\cup R_2$. This is the most complicated case, and we must employ four sub-cases corresponding to what the bottom line evaluates to.\\

\emph{Case 6a)}: $|A| > |L\cup R|$

In this case, the upper line is upper-bounded by $$(|L|+|R|+\min(|R_1|,|R_2|))$$ and the lower line is lower-bounded by $$(|L|+|R|+\min(|L\cup R_1|,|R_2|))$$ The change in EPT-sum is non-positive.\\

\emph{Case 6b)}: $|A| \leq |L\cup R|$ and $|L| < |R_1|$

In this case, the upper line is upper-bounded by $$(|L|+|B|+\min(|R_1|,|R_2|))$$ and the lower line is lower-bounded by $$(|A|+2|L|+\min(|L\cup R_1|,|R_2|))$$ We observe that the inequality $|A|+|L| \geq |B|$ must hold, otherwise $v$ would lead to a more balanced cut than $e$. Therefore, the change in EPT-sum is non-positive.\\

\emph{Case 6c)}: $|A| \leq |L\cup R|$ and $|L| \geq |R_1|$ and $|L\cup R_1| > |R_2|$

In this case, the upper line is upper-bounded by $$(|A|+|R_1|+|R|)$$ and the lower line is equal to $$(|A|+|L|+|R|)$$ By the definition of this case, the change in EPT-sum is non-positive.\\

\emph{Case 6d)}: $|A| \leq |L\cup R|$ and $|L| \geq |R_1|$ and $|L\cup R_1| \leq |R_2|$

In this case, the upper line is upper-bounded by $$(|L|+|B|+|R_1|)$$ and the lower line is equal to $$(|A|+2|L|+2|R_1|)$$ As, again, $|A|+|L| \geq |B|$, the change in EPT-sum is non-positive.

As we have argued, no matter how the sizes of different subtrees relate to each other, $\EPTsum(G,T_{i+1}) \leq \EPTsum(G,T_i)$ and the invariant holds.\\

In the last case, we assume that subtrees $A$ and $B$ lie on the same side of $e$, and assume w.l.o.g. that the components of $G[L(T_i[v])]\setminus e$ are $G[A\cup B\cup R]$ and $G[L]$.\\

\textbf{Case 7:} $e$ lies within $G[L\cup R]$ and $A$ and $B$ are in the same component of $G\setminus e$.\\

\begin{figure}[ht!]
\centering
\includegraphics[width=0.5\textwidth]{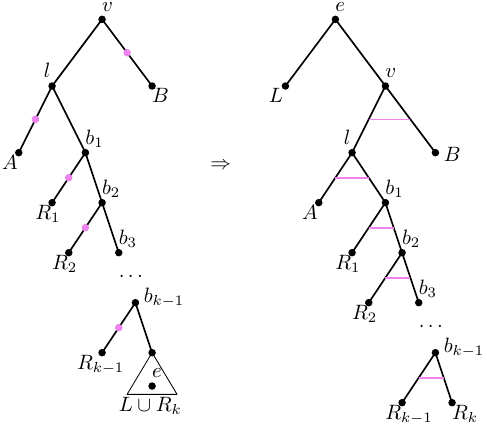}
\caption{Case 7.}\label{fig:case7}
\end{figure}

In this last case, we will look inside the tree covering $L\cup R$ and enumerate each part of $R$, i.e. $R_1,\ldots,R_{k-1}$, that is split off before reaching a cut involving $L$, in order to accurately argue about how the EPT-sum changes when modifying $T$. The crucial point here is that $|L|$ is at least as large as both $|A|$, $|B|$ and $|R_i|$ for every $1\leq i\leq k-1$, otherwise $v$, $l$ or $b_i$ respectively would have led to a more balanced cut than $e$. Looking at Figure \ref{fig:case7}, we reckon that every pink line in $T_{i+1}[e]$ incurs at most as high a EPT-sum as the corresponding pink dot in $T_i[v]$ (clearly, for each $1\leq i<k$, $\min(|R_i|,|L|\sum_{j=i+1}^k |R_j|) \leq |R_i|$; the same argument applies to subdivision nodes going to $A$ and $B$). We therefore disregard the subdivision nodes, and observe that $A$, $B$ and $R_1,\ldots,R_{k-1}$ all increase in depth between $T_i$ and $T_{i+1}$, while $L$ decrease in depth. The change of EPT-sum is therefore upper-bounded by $$(|B|+|A|+\sum_{i=1}^{k-1} |R_i|)-$$ $$(k+1)|L|$$ As all the terms on the upper line, of which there are $k+1$, are at most $|L|$, the change in EPT-sum is non-positive and $\EPTsum(G,T_{i+1}) \leq \EPTsum(G,T_i)$; thus the invariant holds.

With the assumptions we put on the orientation of $T_i[v]$, all possible placements of the balanced edge $e$ are covered by one of the Cases 1-7. The procedure must halt after a finite amount of modifications (specifically, $O(n^2)$ modifications), because each internal node (corresponding to an edge in $G$, of which there are $n-1$) is correctly placed in accordance with $T'$ in $O(n)$ modifications. Specifically, for each subtree, Case 3 may be encountered $O(n)$ times (observe that $e$ has smaller depth in $T_{i+1}$ than in $T_i$ when Case 3 is encountered), while any other case may be encountered maximally one time. From this fact and the fact that the invariant holds, we conclude that the theorem is true.
\end{proof}

\begin{corollary}\label{cor:weights}
Theorem \ref{thm:approx} holds also for vertex-weighted trees.
\end{corollary}
\begin{proof}
This follows directly from replacing sizes of subsets of vertices ($|A|$, $|B|$, $|M|$ etc.) with their weighted equivalents ($\sum_{a\in A} w_a$, $\sum_{b\in B} w_b$ etc.). All the inequalities used in the proof must still hold, since a balanced edge in a vertex-weighted tree is defined to be balanced with respect to the weights.
\end{proof}

\section{Conclusion}

We have shown that the simple balanced cut algorithm gives a 1.5-approximation of EPT-sum on vertex-weighted trees, and therefore of the minimum average search time in trees.\\

There are some questions remaining. For the case of weighted trees, Cicalese et al.~\cite{CJLM14} showed that for the balanced cut algorithm, an approximation ratio of 1.5 must be the best possible, as for any lower ratio, there is an infinite family of graphs failing to achieve that ratio by the balanced cut algorithm. On the other hand, for unweighted trees, we have no such lower bounds. The highest ratio found is by taking the smallest tree in the family found in \cite{HBBPT21}, a tree with 14 vertices; calculating the EPT-sum of the optimal EPT and the one found by taking balanced cuts, the approximation ratio turns out to be $\frac{65}{58}$ for this tree (for the family as a whole, the ratio is $(1+o(1))$). We leave it as an open problem to close the gap between $\frac{65}{58}$ and 1.5 for unweighted trees.\\

It must also be noted that for calculating EPT-sum on unweighted trees, it is still unknown whether the problem is NP-hard or polynomial-time solvable. For the related parameter of edge ranking (equivalent to minimum worst-case search time in trees), there are polynomial- and even linear-time algorithms~\cite{BFN99,de1995optimal,LamYue2001}, but these are non-trivial. If EPT-sum of unweighted trees is solvable in polynomial time, then there are few reasons to believe that an algorithm should be any easier to find.\\

Finally, in this paper we did not focus on edge-weighted trees. It is not hard to show that Lemma \ref{lemma:augmented} still holds for trees with weights on both vertices and edges, by defining the augmented EPT to have a weight on each of their subdivision nodes equal to that of its parent (corresponding to an edge in the input tree). On the other hand, as the sparsest cut of an edge-weighted tree is more complex than the balanced cut of an unweighted (or vertex-weighted) tree, we did not try to adapt the proof of Theorem \ref{thm:approx} to the setting of edge-weighted trees. Of course, the approximation ratio for trees that are both edge- and vertex-weighted cannot be better than that of trees that are only vertex-weighted, but we do not know if there are any strictly larger lower bounds for this case. Cicalese et al.~\cite{CKLPV2016} have looked at this problem and provided hardness results for some very restricted classes of trees, as well as an $O(\frac{\log n}{\log \log n})$-approximation algorithm. It does not contradict their results if the balanced cut algorithm finds a constant approximation also in this case: in fact, for trees that are only edge-weighted, the previously mentioned results of Charikar and Chatziafratis~\cite{CC17} imply that this algorithm should give a constant (8) approximation factor. Whether the actual factor for edge-weighted trees (with or without vertex weights) is as low as 1.5 is an interesting question, and a natural next step for anyone interested in the average performance of searching in trees.

\bibliography{ref}

\end{document}